\documentclass[]{amsart}
\usepackage{enumerate}
\usepackage{amsmath}
\usepackage{amssymb}
\usepackage{amsthm}
\usepackage{setspace}
\usepackage{tikz}
\usepackage{graphicx}

\newtheorem{thm}{Theorem}[section]

\newtheorem{lem}[thm]{Lemma}

\theoremstyle{definition}
\newtheorem{defn}[thm]{Definition}

\theoremstyle{plain}

\newtheoremstyle{nameit}%
{\medskipamount}%
{\medskipamount}%
{\itshape}%
{}%
{\bfseries}%
{.}%
{\labelsep}%
{\thmnote{#3}}%

\theoremstyle{nameit}

\theoremstyle{plain}

\title{The Morphisms with Unstackable Image Words}

\author{C. Robinson Tompkins}

\begin{document}
\maketitle


\begin{abstract}
In an attempt to classify all of the overlap-free morphisms constructively using the Latin-square 
morphism, we came across an interesting counterexample, the Leech square-free morphism. We 
generalize the combinatorial properties of the Leech square-free morphism to gain insights on a
larger class of both overlap-free morphisms and square-free morphisms.
\end{abstract}

\section{Introduction}
The study of overlap-free words and their generators was originated by Axel Thue in 
1912 \cite{thue}. Thue stumbled across overlap-free words in the attempt to find
infinite words that are cube-free. We quickly note that $XXX$ is a cube where 
$X$ is some string of symbols, and a word $W$ avoids cubes if there is no 
subword $XXX$ in $W$. We know the infinite binary word that avoids cubes to be the
Thue-Morse infinite word,
\[01101001100101101001011001101001\ldots\]
\cite{thue}. This infinite word can be generated by function composition of the
Thue-Morse morphism $\mu$ on the letter $0$. Note that the Thue-Morse morphism
is defined as

\vspace{\baselineskip}
\begin{singlespace}
\begin{displaymath}
\mu(t) = \begin{cases}
01 & \textrm{if } t = 0\\
10 & \textrm{if } t = 1.
\end{cases}
\end{displaymath}
\end{singlespace}
\vspace{\baselineskip}
\noindent
Further we define a morphism as a mapping $h:\Sigma^*\to\Delta*$ with 
$\Sigma,\Delta$ being alphabets such that for any two words $V,W\in\Sigma^*$, we 
have $h(VW) = h(V)h(W)$. A morphism $h$ is called cube free provided $h(W)$
is cube-free if and only if $W\in\Sigma^*$ is also cube-free.

Our primary concern however is dealing with overlaps instead of cubes. An overlap is the pattern $cXcXc$ where
$c$ represents a single letter and $X$ is a word with possibly zero letters. The standard example of a word 
that is an overlap in its entirety is ``alfalfa'', and an overlap-free word is a word in which no overlap
occurs. 

A morphism  $h$ is said to be overlap free so long as we have $X\in\Sigma^*$ overlap-free if and only if 
$h(X)$ is overlap-free.
Surprisingly it is known that $\mu$ and its natural complement are the only non-trivial overlap-free morphisms
on the two letter alphabet $\{-0,1\}$ \cite{seebold_binary}.

In the early 80's Crochemore, Ehrenfeucht, and Rozenberg made substancial progress towards 
classifying the square-free morphisms \cite{crochemore_squarefree},\cite{ehren}. Further in 2004,
Richomme and Wlazinski published a result classifying all overlap-free morphisms 
\cite{ric_testsets}. However, their result much like the results of Crochemore, Ehrenfeucht, and 
Rozenberg rely on test-sets of words for the morphism in question. Furthermore, the tests for
Richomme and Wlazinski grow factorially with the size of the input alphabet.

Since the late 1990's and early 2000's, several results have surfaced pursuing a constructive
understanding of the class of overlap-free morphisms. In 2001, Frid suggested using the 
structure of the cyclic group of order $n$ to define each image word accordingly for a morphism 
on an alphabet
of $n$ letters \cite{frid}. In 2007 we extended Frid's result to the use of the Latin-square 
structure to
define our morphism structure \cite{tompkins}.

In a vain attempt to use the Latin-square morphism construction to classify all of the overlap-free
morphisms, we stumbled across the Leech square-free morphism in \cite{albook}. The following is the 
Leech square-free morphism

\begin{singlespace}
\begin{displaymath}
h(t) = \begin{cases}
0121021201210 &\textrm{for } t = 0\\
1202102012021&\textrm{for } t = 1\\
2010210120102&\textrm{for } t = 2.
\end{cases}
\end{displaymath}
\end{singlespace}

\quad\\
\noindent
which originally appeared in \cite{leech}.  This lead us to the definition of the morphism with unstackable image words. Note that the definition depends upon a combinatorial property and is not entirely constructive. We have yet to overcome this problem.


\section{Preliminaries}

We will use the standard definitions from the Lothaire book on combinatorics on words for our definitions 
with a few additions \cite{lothare_book}. 

We begin by defining an alphabet $\Sigma$ to be a finite set of symbols from which we will make words 
by concatenation (note, we will use capital Greek letters for alphabets). Further, we define a word $W$ to 
be a list of symbols from any alphabet $\Sigma$ written horizontally (we will use capital letters to denote 
words and lower case letters to denote letters). We will denote the word with no letters, that is the empty 
word, by $\varepsilon$. 

\subsection{Words}

The length (or number of letters) for a word $W$ will be written $|W|$. Note that we will use the same 
symbol to represent the size of a set or absolute value. The difference will be clear based on context. 
Notice that $|\varepsilon| = 0$. Further we will represent $|W|_a$ to represent the number of times the 
letter $a$ occurs in $W$. Also we will use $|W|_{aba}$ to represent the number of times the word $aba$ 
occurs in $W$. For example if $C = abaababa$, then we have $|C|_{aba} = 3$ along with $|C| = 8$.

A word $U$ is a factor of a word $V$ if there exist two (possibly empty) words $S$ and $T$ such that $V = 
TUS$. We will also say that $U$ is a subword of $V$ (or $V$ contains $U$). If $T=\varepsilon$, then we 
call $U$ the prefix of $V$. Similarly, if $S = \varepsilon$, then we call $U$ the suffix of $V$.

For some alphabet $\Sigma$, $\Sigma^*$ is the Kleene closure of our alphabet. That is, $\Sigma^*$ is all 
of the possible words over the alphabet $\Sigma$. Notice that $\Sigma^*$ is the free monoid over the set 
$\Sigma$.

\subsection{Morphisms}

A morphism $h$ is a mapping from $\Sigma^*$ into $\Delta^*$, where $\Sigma$ and $\Delta$ are 
alphabets, such that $h(WV) = h(W)h(V)$ for all 
words $W,V\in\Sigma^*$, and $h(\varepsilon) = \varepsilon$. Note that $W$ and $V$ could 
potentially be single letters. We note that if $X\subseteq\Sigma$ ($X$ represents a set of words) for some 
alphabet $\Sigma$, $h(X)$ represents the set of words $\{h(W):W\in X\}$. Further, we call $h$ 
non-erasing if for all $a\in\Sigma$, where $\Sigma$ is an alphabet, $h(a)\neq\varepsilon$.

Recall from earlier that the Thue-Morse morphism, $\mu$ defined as

\vspace{\baselineskip}
\begin{singlespace}
\begin{displaymath}
\mu(t) = 
\begin{cases}
01, & \textrm{for } t=0\\
10, & \textrm{for } t=1,
\end{cases}
\end{displaymath}
\end{singlespace}
\vspace{\baselineskip}

\noindent
is a morphism defined on the alphabet with two letters. For convenience, we will call the alphabet with $n
$ letters $\Sigma_n$. Infinite words are possible with such a morphism. We have displayed the $n^
\mathrm{th}$ Thue-Morse word as being $\mu^n(0)$. We will use $\omega$ to represent the first infinite 
ordinal. So the Thue-Morse infinite word becomes
\[\mathbf{T} = \lim_{n\to\infty}\mu^n(0) = \mu^\omega (0),\]
 as previously seen. Note that we will use bold capitol letters to represent infinite words, with $\mathbf{T}$ 
here representing 
the Thue-Morse infinite word.

When discussing $\Sigma_2 = \{0,1\}$, the two letter alphabet we will use $\bar{0}$ to denote 
the complement of 0 (or 1 if we need that complement). That is $\bar{0}=1$ and $\bar{1}=0$. This will 
become necessary in Chapter 2.

For some morphism $h:\Sigma^*\to\Delta^*$, we will call $h$ uniform if $|h(a)| = n$ for some integer $n$ 
for all $a\in\Sigma$ (more exactly, in this case we will call $h$ $n$-uniform).  We will call a morphism $h:
\Sigma^*\to\Delta^*$ square-free when $h(W)$ is square-free if and only if $X\in\Sigma^*$ is square-free. 
Similarly, we will call $h:\Sigma^*\to\Delta^*$ an overlap-free morphism when $h(W)$ is overlap-free if 
an only if $W\in\Sigma^*$ is overlap-free.


\section{The Morphism With Unstackable Image Words}

In a vain attempt to classify all of the overlap-free morphisms using the latin square morphism 
\cite{tompkins}, we stumbled across the Leech square-free morphism in \cite{albook}. The following is the 
Leech square-free morphism

\begin{singlespace}
\begin{equation*}
h(t) = \begin{cases}
0121021201210 & \textrm{for } t = 0\\
1202102012021 & \textrm{for } t = 1\\
2010210120102 & \textrm{for } t = 2,
\end{cases}
\end{equation*}
\end{singlespace}

\quad\\
\noindent
which originally appeared in \cite{leech}. Noticing that this morphism was overlap-free put a hole in our 
attempt to classify all of the overlap-free morphisms using Latin square morphisms. But on the other 
hand, we now could potentially find another class of overlap-free morphisms that could be explained in a 
better manner than with test-sets as in \cite{ric_testsets}.

Using the test-set result given by Richomme and Wlazinski, we found the following overlap-free 
morphisms on four letters

\begin{singlespace}
\begin{equation*}
f(t) = \begin{cases}
0123 1230 1 0321 3210 & \textrm{for } t = 0\\
1230 2301 2 1032 0321 & \textrm{for } t = 1\\
2301 3012 3 2103 1032 & \textrm{for } t = 2\\
3012 0123 0 3210 2103 & \textrm{for } t = 3,
\end{cases}
\end{equation*}
\end{singlespace}

\quad\\
\noindent
and

\begin{singlespace}
\begin{equation*}
g(t) = \begin{cases}
0123 0122121120 3210 & \textrm{for } t = 0\\
1230 1300303301 0321 & \textrm{for } t = 1\\
2301 2012331022 1032 & \textrm{for } t = 2\\
3012 3011010013 2103 & \textrm{for } t = 3.
\end{cases}
\end{equation*}
\end{singlespace}

\quad\\
\noindent
The morphism $g$ raised a considerable number of questions as to why it was overlap-free. It seemed to 
avoid a considerable number of the techniques used in the proof For the Latin square morphisms. So the 
natural question was: what does the morphism $g$ have in common with the Leech square-free 
morphism that causes its overlap-freeness. 


\section{Definitions and Theorems}

The overlap-free morphisms displayed above are tied together with the following definition.

\begin{defn}
Let $h:\Sigma^*\to\Delta^*$ be an $n$-uniform morphism. We say that $h$ is a morphism with
unstackable image words if it 
satisfies the 
following properties:
\begin{itemize}
\item[(i)] $h(W)$ is overlap-free for all overlap-free words $W\in\Sigma^*$ with $|W|=3$.
\item[(ii)] For $a,b\in\Sigma$, and for all $V\in\Sigma^*$ such that $|V|\le\lfloor n/2\rfloor$,
\[h(a) = SV\quad\textrm{and}\quad h(b) = VU\]
if and only if $S$ is not a suffix of any image word of $h$ and $U$ is not a prefix of any image word of $h
$.
\end{itemize}
\label{def_5_1}
\end{defn}

We now prove a lemma that captures the combinatorial properties in the first 
portion of Definition \ref{def_5_1}.

\begin{lem}
Let $\Sigma$ be an alphabet with more than one letter. Let $h:\Sigma^*\to\Delta^*$ be a morphism such 
that $h(W)$ is overlap-free for all overlap-free $W\in
\Sigma^*$ with $|W| = 3$. We then have the following properties:
\begin{itemize}
\item[(i)] $h(a)$ is overlap-free for all $a\in\Sigma$.
\item[(ii)] $h(a)h(b)$ is overlap-free for all $a,b\in\Sigma$.
\item[(iii)] $h(a)$ and $h(b)$ do not begin or end with the same letter, whenever $a,b\in\Sigma$ and $a
\neq b$.
\end{itemize}
\label{lem_5_2}
\end{lem}

\begin{proof}
(i) Let us first state that the result does not apply when $|\Sigma|=1$ because there are no overlap-free 
word of length three for this alphabet. For $|\Sigma|>1$ this result is clear because if we assume for a 
contradiction that $h(a)$ contained an overlap for 
any $a\in\Sigma$, then $h(bab)$, with $b\neq a$, would contain an overlap which contradicts our 
assumption.

(ii) Similar to (i), if we assume that $h(a)h(b) =  h(ab)$ contained an overlap, then $h(aba)$ would contain 
an overlap. Again this contradicts our assumption. We also must show that $h(aa)$ does not contain an 
overlap. Assume for a contradiction that it does, and we quickly obtain our contradiction by observing that 
then $h(aab)$ must contain an overlap.

(iii) Assume that for some $a,b\in\Sigma$, $h(a)$ and $h(b)$ begin with the same letter. Then, 
$h(aab)$ would contain an overlap, which contradicts our assumption. The argument for $h(b)$ and 
$h(b)$ ending with different letters is similar.
\end{proof}

\begin{thm}
Any morphism with unstackable image words is overlap-free.
\label{thm_5_3}
\end{thm}

\begin{proof}
We begin by assuming that $h$ is a morphism with unstackable image words with 
$|h(a)|=n$ for all $a\in\Sigma$. We must show 
that for all $W\in\Sigma^*$, $W$ 
is overlap-free if and only if $h(W)$ is overlap-free. We will begin with the easy direction first.

\subsection{The $\Leftarrow$ direction} Assume that $W=AcXcXcB$, so that we can argue by 
contrapositive that $h(W)$ must also contain an 
overlap. Notice that
\[h(W) = h(A)h(c)h(X)h(c)h(X)h(c)h(B).\]
Set $h(c) = dY$ where $d\in\Sigma$ and $Y\in\Sigma^*$, then $h(W) = h(A)dYh(X)dYh(x)dYh(B)$. 
So then $h(W)$ contains the overlap $dYh(X)dYh(X)d$, and we are done with the first portion of our 
argument.

\subsection{The $\Rightarrow$ direction}Conversely we will argue by contrapositive. We will assume that 
$h(W)$ contains an overlap and show 
that $W$ must also contain an overlap. So assume that for some $W\in\Sigma^*$ we have
\begin{displaymath}
h(W) = Ac_{j_0}Xc_{j_1}Xc_{j_2}B,
\end{displaymath}
where $c=c_{j_0}=c_{j_1}=c_{j_2}$. We use the 0, 1 and 2 to denote which $c$ we will refer to. Further, 
the index $j_i$ will refer to which letter in the word $h(W)$ we are referring to, noting that we are indexing 
beginning with 0.

We will proceed with two separate arguments. The first argument will be that it is not possible to write 
$h(W)$ with $|cX|\not\equiv 0\pmod{n}$. The second argument will be that $W$ must contain an overlap 
if $|cX|\equiv 0 \pmod{n}$.

\subsubsection{The $|cX|\not\equiv 0\pmod{n}$ case.} Notice that we must have the overlap in $h(W)$ 
contained 
in $h(Z)$ where $|Z|>3$ is some subword of $W$. Otherwise we would be breaking hypothesis (i) in the 
definition of Pooh morphisms. 

We begin by setting
\[r_i\equiv j_i\pmod{n},\]
where $i\in\{0,1,2\}$ and $r_i\in\{0,1,\ldots,n-1\}$. We will argue first based on the number of tiles that the 
overlap occurs in, and then by cases. When the overlap occurs over four tiles (noting that occurring over 
three tiles contradicts the hypothesis), we will observe four cases. The cases are

\begin{singlespace}
\begin{align*}
r_0\le r_2<r_1,\\
r_2<r_0<r_1,\\
r_1<r_0 \le r_2,\\
r_1<r_2<r_0.\\
\end{align*}
\end{singlespace}

\noindent
We note that the cases $r_0<r_1<r_2$ and $r_2<r_1<r_0$ force the overlap to occur in a number other 
than four tiles. When the overlap occurs in more than four tiles we will more simply consider the two 
cases $r_0<r_1$ and $r_1<r_0$. Finally, we note the following relationship between $r_0$, $r_1$, and 
$r_2$.
\begin{equation}
r_2\equiv 2r_1 - r_0 \pmod{n}.
\label{eq_5_1}
\end{equation}

Consider the notion of the tiling of a line segment. We will use this notion of tiling in application to working 
with $h(W)$. The tiles we speak of are the image words of $h$. Note that all the image words must be of 
the 
same length $n$, this is crucial to our argument. For ease we will use $T_{s_i}$ with $i \in\{0,1,2\}$ to 
denote the tile containing $c_{j_i}$. Note that $s_i$ is the number of the tile if we numbered them 
starting with the first tile as $T_{0}$.

\emph{The overlap is contained in 4 tiles.} Let us consider the case where there is some subword of $W$, 
say $Z$, with $|Z|=4$ and the overlap in $h(W)$ is contained in $h(Z)$. As in the argument for a Latin 
square morphism to be overlap-free we will consider the word $h(Z)$ to be a line. We will draw in small 
vertical lines to signify the edges of the tiles, and we will draw taller labeled vertical lines to signify the $c
$'s in the overlap.

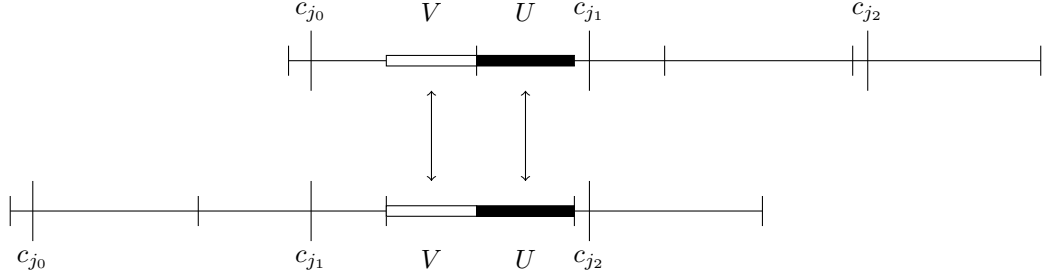
\begin{figure}[h]
\centerline{
\begin{tikzpicture}
%
%
\draw (-3.1,1.0)--(6.9,1.0);
\draw (-6.8,-1.0)--(3.2,-1.0);
\draw (-6.8,-1.2)--(-6.8,-.8);
\draw (-4.3,-1.2)--(-4.3,-.8);
\draw (-1.8,-1.2)--(-1.8,-.8);
\draw (.7,-1.2)--(.7,-.8);
\draw (3.2,-1.2)--(3.2,-.8);
\draw (-3.1,.8)--(-3.1,1.2);
\draw (-.6,.8)--(-.6,1.2);
\draw (1.9,.8)--(1.9,1.2);
\draw (4.4,.8)--(4.4,1.2);
\draw (6.9,.8)--(6.9,1.2);
\draw (-2.8,.6)--(-2.8,1.4);
\draw (.9,.6)--(.9,1.4);
\draw (4.6,.6)--(4.6,1.4);
\draw (-2.8,-.6)--(-2.8,-1.4);
\draw (.9,-.6)--(.9,-1.4);
\draw (-6.5,-.6)--(-6.5,-1.4);
\coordinate[label=above:$c_{j_0}$] (A) at (-2.8,1.4);
\coordinate[label=above:$c_{j_1}$] (B) at (.9,1.4);
\coordinate[label=above:$c_{j_2}$] (C) at (4.6,1.4);
\coordinate[label=below:$c_{j_0}$] (a) at  (-6.5,-1.4);
\coordinate[label=below:$c_{j_1}$] (b) at (-2.8,-1.4);
\coordinate[label=below:$c_{j_2}$] (c) at  (.9,-1.4);
\draw[fill=white] (-1.8,1.07) rectangle (-.6,.93);
\draw[fill=white] (-1.8,-.93) rectangle (-.6,-1.07);
\draw[<->](-1.2,-.6)--(-1.2,.6) ;
\coordinate[label=above:$V$] (Vu) at (-1.2,1.4);
\coordinate[label=below:$V$] (vd) at (-1.2,-1.4);
\draw[fill=black] (-.6,1.07) rectangle (.7,.93);
\draw[fill=black] (-.6,-.93) rectangle (.7,-1.07);
\coordinate[label=above:$U$] (Uu) at (.05,1.4);
\coordinate[label=below:$U$] (Ud) at (.05,-1.4);
\draw[<->](.05,.6)--(.05,-.6);
\end{tikzpicture}
}
\caption{The short overlap with $r_2<r_0<r_1$}\label{fig_5_1}
\end{figure}

In Figure \ref{fig_5_1}, we have taken $c_{j_0}Xc_{j_1}Xc_{j_2}$ and written it twice aligning $c_{j_0}Xc_{j_1}$ 
in the upper line with $c_{j_1}Xc_{j_2}$ in the lower line for the purpose of equating the terms through 
the overlap. Figure \ref{fig_5_1} displays the case when $r_2<r_0<r_1$. We remark here that the case when 
$r_2=r_0<r_1$ proceeds in the same manner.

Let $V$ to be the final $r_1-r_0$ letters in the tile $T_{s_0}$, as we have drawn in Figure \ref{fig_5_1}.  Similarly 
we choose $U$ to be the first $r_1-r_2$ letters in $T_{s_1}$. Now equation (\ref{eq_5_1}) gives that in 
the $r_2<r_0<r_1$ situation we have that $n-(r_1-r_0) = r_1-r_2$. Clearly then we must have $|U|=r_1-
r_2\le\lfloor n/2 \rfloor$ or $|V|=r_1-r_0\le\lfloor n/2 \rfloor$. In the case when $|V|\le\lfloor n/2\rfloor$ we 
cannot equate $U$ with any prefix of a tile which leads to a contradiction. In the other case when $|U|\le
\lfloor n/2\rfloor$ we cannot equate $V$ with any suffix of a tile which leads to a contradiction. So this 
case is not possible.

We now consider the case where $r_0<r_2<r_1$ as shown in Figure \ref{fig_5_a}.

\begin{figure}[h]
\centerline{
\begin{tikzpicture}
%
%
\draw (-4.5,1.0)--(5.5,1.0);
\draw (-8.7,-1.0)--(1.3,-1.0);
\draw (-4.5,.8)--(-4.5,1.2);
\draw (-2.0,.8)--(-2.0,1.2);
\draw (.5,.8)--(.5,1.2);
\draw (3.0,.8)--(3.0,1.2);
\draw (5.5,.8)--(5.5,1.2);
\draw (-8.7,-.8)--(-8.7,-1.2);
\draw (-6.2,-.8)--(-6.2,-1.2);
\draw (-3.7,-.8)--(-3.7,-1.2);
\draw (-1.2,-.8)--(-1.2,-1.2);
\draw (1.3,-.8)--(1.3,-1.2);
\coordinate[label=above:$c_{j_0}$] (A) at (-4.2,1.4);
\coordinate[label=above:$c_{j_1}$] (B) at (0,1.4);  
\coordinate[label=above:$c_{j_2}$] (C) at (4.2,1.4); 
\coordinate[label=below:$c_{j_0}$] (a) at  (-8.4,-1.4);
\coordinate[label=below:$c_{j_1}$] (b) at (-4.2,-1.4); 
\coordinate[label=below:$c_{j_2}$] (c) at  (0,-1.4);  
\draw (-4.2,.6)--(-4.2,1.4);
\draw (0,.6)--(0,1.4);
\draw (4.2,.6)--(4.2,1.4);
\draw (-8.4,-1.4)--(-8.4,-.6);
\draw (-4.2,-1.4)--(-4.2,-.6);
\draw (0,-1.4)--(0,-.6);
\draw[fill=black] (-2,1.07) rectangle (-1.2,.93);
\draw[fill=black] (-2,-.93) rectangle (-1.2,-1.07);
\coordinate[label=above:$U$] (Vu) at (-1.6,1.4);
\coordinate[label=below:$U$] (vd) at (-1.6,-1.4);
\draw[fill=white] (-3.7,1.07) rectangle (-2.0,.93);
\draw[fill=white] (-3.7,-.93) rectangle (-2.0,-1.07);
\coordinate[label=above:$V$] (Uu) at (-2.8,1.4);
\coordinate[label=below:$V$] (Ud) at (-2.8,-1.4);
\draw[<->](-1.6,-.6)--(-1.6,.6);
\draw[<->](-2.85,-.6)--(-2.85,.6);
\end{tikzpicture}
}
\caption{The short overlap with $r_0<r_2<r_1$}\label{fig_5_a}
\end{figure}
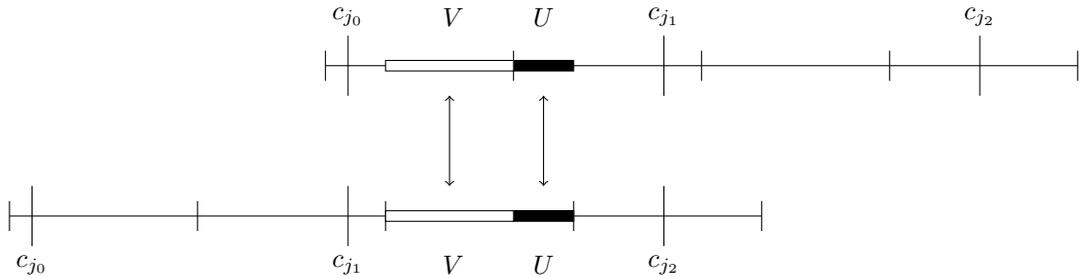

\noindent
In this case we choose $V$ to be the final $r_1-r_0$ letters in $T_{s_0}$, and we also pick $U$ to be the 
first $r_1-r_2$ letters in $T_{s_1}$ as drawn in Figure \ref{fig_5_a}. Again we notice that $n-(r_1-r_0) = r_1-r_2$ 
so either $V\le\lfloor n/2\rfloor$ or $|U|\le\lfloor n/2\rfloor$, either of which is impossible. So we cannot 
have this case occurring either.

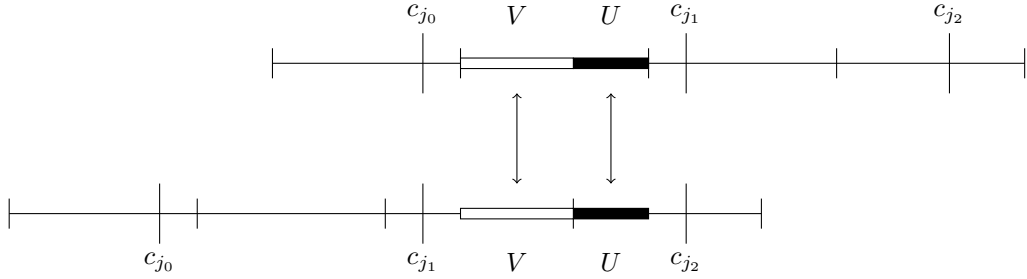
\begin{figure}[h]
\centerline{
\begin{tikzpicture}
%
%
\draw (-6.5,-1.0)--(3.5,-1.0);
\draw (-3.0,1.0)--(7.0,1.0);
\draw (-3.0,.8)--(-3.0,1.2);
\draw (-.5,.8)--(-.5,1.2);
\draw (2.0,.8)--(2.0,1.2);
\draw (4.5,.8)--(4.5,1.2);
\draw (7.0,.8)--(7.0,1.2);
\draw (-6.5,-1.2)--(-6.5,-.8);
\draw (-4.0,-1.2)--(-4.0,-.8);
\draw (-1.5,-1.2)--(-1.5,-.8);
\draw (1.0,-1.2)--(1.0,-.8);
\draw (3.5,-1.2)--(3.5,-.8);
\draw (-1.0,.6)--(-1.0,1.4);
\draw (2.5,.6)--(2.5,1.4);
\draw (6.0,.6)--(6.0,1.4);
\draw (-4.5,-1.4)--(-4.5,-.6);
\draw (-1.0,-1.4)--(-1.0,-.6);
\draw (2.5,-1.4)--(2.5,-.6);
\coordinate[label=above:$c_{j_0}$] (A) at (-1.0,1.4);
\coordinate[label=above:$c_{j_1}$] (B) at (2.5,1.4);  
\coordinate[label=above:$c_{j_2}$] (C) at (6.0,1.4); 
\coordinate[label=below:$c_{j_0}$] (a) at  (-4.5,-1.4);
\coordinate[label=below:$c_{j_1}$] (b) at (-1.0,-1.4); 
\coordinate[label=below:$c_{j_2}$] (c) at  (2.5,-1.4); 
\draw[fill=black] (1.0,1.07) rectangle (2,.93);
\draw[fill=black] (1.0,-.93) rectangle (2,-1.07);
\coordinate[label=above:$U$] (Vu) at (1.5,1.4);
\coordinate[label=below:$U$] (vd) at (1.5,-1.4);
\draw[fill=white] (-.5,1.07) rectangle (1,.93);
\draw[fill=white] (-.5,-.93) rectangle (1,-1.07);
\coordinate[label=above:$V$] (Uu) at (.25,1.4);
\coordinate[label=below:$V$] (Ud) at (.25,-1.4);
\draw[<->](1.5,-.6)--(1.5,.6);
\draw[<->](.25,-.6)--(.25,.6);
\end{tikzpicture}
}
\caption{The short overlap with $r_1<r_2<r_0$}\label{fig_5_2}
\end{figure}

We now consider the cases with $r_1<r_2\le r_0$ and $r_1<r_0<r_2$. Figure \ref{fig_5_2} gives the situation 
when $r_1<r_2<r_0$ (note that the case when $r_1<r_2=r_0$ is similar, and the same applies to the 
arguments above). Notice that in both of these cases we have that $n-(r_2-r_1) = r_0-r_1$.

For the case depicted in Figure \ref{fig_5_2}, $r_1<r_2<r_0$, we assume that $V$ is the final $r_0-r_1$ letters in 
$T_{s_1}$, and we also assume that $U$ is the first $r_2-r_1$ letters in $T_{s_2}$. Now either $|U|\le 
\lfloor n/
2\rfloor$ or $|V|\le\lfloor n/2\rfloor$. In either case we have a contradiction.

Now we consider the case where $r_1<r_0<r_2$, which is displayed in Figure \ref{fig_5_b}.
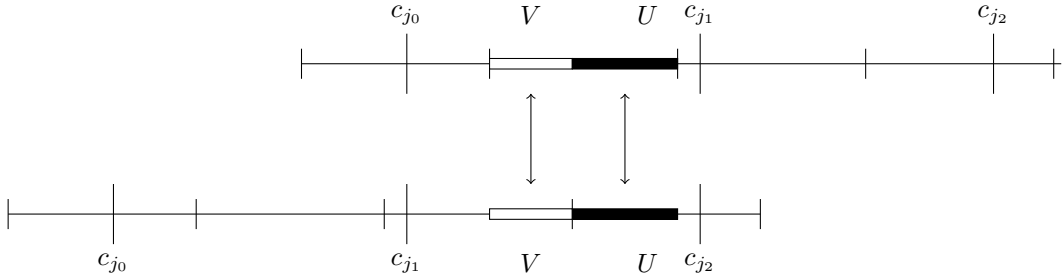
\begin{figure}[h]
\centerline{
\begin{tikzpicture}
%
%
\draw (-5.3,1.0)--(4.8,1.0);
\draw (-9.2,-1.0)--(.8,-1.0);
\draw (-5.3,.8)--(-5.3,1.2);
\draw (-2.8,.8)--(-2.8,1.2);
\draw (-.3,.8)--(-.3,1.2);
\draw (2.2,.8)--(2.2,1.2);
\draw (4.7,.8)--(4.7,1.2);
\draw (-9.2,-1.2)--(-9.2,-.8);
\draw (-6.7,-1.2)--(-6.7,-.8);
\draw (-4.2,-1.2)--(-4.2,-.8);
\draw (-1.7,-1.2)--(-1.7,-.8);
\draw (.8,-1.2)--(.8,-.8);
\draw (-3.9,.6)--(-3.9,1.4);
\draw (0,.6)--(0,1.4);
\draw (3.9,.6)--(3.9,1.4);
\coordinate[label=above:$c_{j_0}$] (A) at (-3.9,1.4);
\coordinate[label=above:$c_{j_1}$] (B) at (0,1.4);  
\coordinate[label=above:$c_{j_2}$] (C) at (3.9,1.4); 
\coordinate[label=below:$c_{j_0}$] (a) at  (-7.8,-1.4);
\coordinate[label=below:$c_{j_1}$] (b) at (-3.9,-1.4); 
\coordinate[label=below:$c_{j_2}$] (c) at  (0,-1.4);  
\draw (-7.8,-1.4)--(-7.8,-.6);
\draw (-3.9,-1.4)--(-3.9,-.6);
\draw (0,-1.4)--(0,-.6);
\draw[fill=white] (-2.8,1.07) rectangle (-1.7,.93);
\draw[fill=white] (-1.7,-.93) rectangle (-2.8,-1.07);
\coordinate[label=above:$V$] (Vu) at (-2.25,1.4);
\coordinate[label=below:$V$] (vd) at (-2.25,-1.4);
\draw[fill=black] (-1.7,1.07) rectangle (-.3,.93);
\draw[fill=black] (-.3,-.93) rectangle (-1.7,-1.07);
\coordinate[label=above:$U$] (Uu) at (-.7,1.4);
\coordinate[label=below:$U$] (Ud) at (-.7,-1.4);
\draw[<->](-2.25,.6)--(-2.25,-.6);
\draw[<->](-1.0,.6)--(-1.0,-.6);
\end{tikzpicture}
}
\caption{The short overlap with $r_1<r_0<r_2$}\label{fig_5_b}
\end{figure}

In the case displayed here in Figure \ref{fig_5_b}, we again assume that $V$ occurs in the final $r_0-r_1$ letters 
of $T_{s_1}$, and we also assume that $U$ occurs in the first $r_2-r_1$ letters of $T_{s_2}$. We then 
have that either $|V|\le\lfloor n/2\rfloor$ or that $|U|\le\lfloor n/2\rfloor$. Either case is a contradiction. So 
we cannot have our overlap occurring in four tiles. Thus, we consider the case when the overlap occurs 
in more than four tiles.

We also not that if we are in the case when the overlap occurs in five tiles, the same arguments hold.

\emph{The overlap is contained in more than four tiles.} We will look at the cases with $r_0<r_1$ and 
$r_1<r_0$, and we will only look at the beginning of the overlap. So we consider Figure \ref{fig_5_3} for the case 
when $r_0<r_1$.

\begin{figure}[h]
\centerline{
\begin{tikzpicture}
\draw (-6.0,-1.0)--(3.5,-1.0);
\draw (-4.0,1.0)--(2.5,1.0);
\coordinate[label=right:$\bullet$] (e) at (3,0);
\coordinate[label=right:$\bullet$] (f) at (3.5,0);
\coordinate[label=right:$\bullet$] (g) at (4,0);
\draw (-4.0,.8)--(-4.0,1.2);
\draw (-1.0,.8)--(-1.0,1.2);
\draw (2.0,.8)--(2.0,1.2);
\draw (-6.0,-.8)--(-6.0,-1.2);
\draw (-3.0,-.8)--(-3.0,-1.2);
\draw (0,-.8)--(0,-1.2);
\draw (3.0,-.8)--(3.0,-1.2);
\coordinate[label=above:$c_{j_0}$] (A) at (-3.5,1.4);
\draw (-3.5,.6)--(-3.5,1.4);
\coordinate[label=below:$c_{j_1}$] (a) at (-3.5,-1.4);
\draw (-3.5,-.6)--(-3.5,-1.4);
\draw[fill=white] (-3,1.07) rectangle (-1,.93);
\draw[fill=white] (-1,-.93) rectangle (-3,-1.07);
\coordinate[label=above:$V$] (Vu) at (-2,1.4);
\coordinate[label=below:$V$] (Vd) at (-2,-1.4);
\draw[<->](-2.0,-.6)--(-2.0,.6);
\draw[fill=black] (-1,1.07) rectangle (0,.93);
\draw[fill=black] (0,-.93) rectangle (-1,-1.07);
\coordinate[label=above:$U$] (Uu) at (-.5, 1.4);
\coordinate[label=below:$U$] (Ud) at (-.5,-1.4);
\draw[<->](-.5,.6)--(-.5,-.6);
\draw[<-](1.0,-.6)--(1.0,.6);
\draw[-](.9,-.1)--(1.1,.1);
\end{tikzpicture}
}
\caption{The long overlap with $r_0<r_1$ and $r_1-r_0\ge\lfloor n/2\rfloor$}\label{fig_5_3}
\end{figure}
We will consider the case with $r_1-r_0\ge\lfloor n/2\rfloor$, as we will cover the logic behind the 
argument for $r_1-r_0\le\lfloor n/2\rfloor$ in Figure \ref{fig_5_4}.

Let $V$ be the final $r_1-r_0$ letters in $T_{s_0}$, then equating yields $V$ as the beginning $r_0-r_1$ 
letters of $T_{s_1+1}$. Since $|V|\ge \lfloor n/2\rfloor$ we can equate the suffix of $T_{s_1+1}$, call it $U$ 
(which is labeled with a dotted line in Figure \ref{fig_5_3}),
with $T_{s_0+1}$. So we have $T_{s_1+1} = VU$. Similarly we can set $S\in\Delta^*$ such that 
$T_{s_0+1} = US$. Notice now that $|U| = n-(r_1-r_0)\le\lfloor n/2\rfloor$. Thus $S$ cannot begin any 
image word of $h$ so the overlap is impossible.

A note for the case when $r_1-r_0\le\lfloor n/2\rfloor$. In this case $|V|\le\lfloor n/2\rfloor$ and we would 
not be able to equate $U$.

\begin{figure}[h]
\centerline{
\begin{tikzpicture}
\draw (-6.0,1.0)--(.5,1.0) ;
\draw (-4.5,-1.0)--(2.0,-1.0) ;
\coordinate[label=right:$\bullet$] (e) at (1.5,0);
\coordinate[label=right:$\bullet$] (f) at (2,0);
\coordinate[label=right:$\bullet$] (g) at (2.5,0);
\draw (-6.0,.8)--(-6.0,1.2) ;
\draw (-3.0,.8)--(-3.0,1.2) ;
\draw (0,.8)--(0,1.2);
\draw (-4.5,-1.2)--(-4.5,-.8) ;
\draw (-1.5,-1.2)--(-1.5,-.8) ;
\draw (1.5,-1.2)--(1.5,-.8);
\draw (-4.0,.6)--(-4.0,1.4) ;
\coordinate[label=above:$c_{j_0}$] (A) at (-4.0,1.4);
\draw (-4.0,-1.4)--(-4.0,-.6) ;
\coordinate[label=below:$c_{j_1}$] (b) at (-4.0,-1.4);
\draw[fill=white] (-3,1.07) rectangle (-1.5,.93);
\draw[fill=white] (-3,-.93) rectangle (-1.5,-1.07);
\coordinate[label=above:$V$] (Vu) at (-2.25,1.4);
\coordinate[label=below:$V$] (Vd) at (-2.25,-1.4);
\draw[<->](-2.25,-.6)--(-2.25,.6);
\draw[<->](-.75,-.6)--(-.75,.6);
\draw (-.85,-.1)--(-.65,.1);
\end{tikzpicture}
}
\caption{The long overlap with $r_1<r_0$ and $r_0-r_1\le\lfloor n/2\rfloor$}\label{fig_5_4}
\end{figure}
Figure \ref{fig_5_4} gives the case when $r_1<r_0$ with $r_0-r_1\le\lfloor n/2\rfloor$. In a similar manner to the 
case where $r_0<r_1$ we pick $V$ to be the suffix of $r_0-r_1$ letters in $T_{s_1}$. Now we can 
possibly find an image word $T_{s_0+1} = VU$ for some $U\in\Delta^*$. But because $|V| = r_0-r_1\le
\lfloor n/2\rfloor$, $U$ cannot be the prefix of any image word, so this formulation of the overlap is 
impossible.

So we see that in order for $h(W)$ to contain an overlap, it must be one so that $|cX|\equiv 0 \pmod{n}$.

\subsubsection{The $|cX|\equiv 0 \pmod{n}$ case.} From Lemma \ref{lem_5_2} we know that the 
beginning letters and ending letter for each image word in $h$ must be distinct. Further we know that the 
suffix of $T_{s_1}$ must be identical to the suffix of $T_{s_0}$ as $r_0=r_1$. This implies that 
$T_{s_0}=T_{s_1}$. Similarly $T_{s_1} = T_{s_2}$.

Pick $d$ to be the letter such that $h(d) = T_{s_0} = T_{s_1} = T_{s_2} = T$. Further because
\[h(W) = Ac_{j_0}Xc_{j_1}Xc_{j_2}XB,\]
we can find subwords $C,D,Y$ of $W$ so that
\[h(W) = h(CdYdYdD) = AcXcXcB.\]
Now we must have that $W = CdYdYdD$ which contains an overlap. Thus we are done.

\end{proof}


\section{The Square-Free Adaptation}

Similarly to the definition of the overlap-free morphisms with unstackable image words we can define 
square-free morphisms with unstackable image words in the 
following manner.

\begin{defn}
Let $h:\Sigma^*\to\Delta^*$ be an $n$-uniform morphism. We call $h$ a square-free morphism with 
unstackable image words if it 
satisfies the following properties:
\begin{itemize}
\item[(i)] $h(W)$ is square-free for all square-free words $W\in\Sigma^*$ with $|W| = 3$
\item[(ii)] $h(a)$ and $h(b)$ do not begin or end with the same letter for all $a,b\in\Sigma$ with $a\neq b$.
\item[(iii)] For $a,b\in\Sigma$, and for all $V\in\Sigma^*$ such that $|V|\le\lfloor n/2\rfloor$,
\[h(a) = SV\quad\textrm{and}\quad h(b) = VU\]
if and only if $S$ is not a suffix of any image word of $h$ and $U$ is not a prefix of any image word of $h
$.
\end{itemize}
\label{defn_5_4}
\end{defn}

Because we cannot consider words like $aab$ to put into $h$, we must add property (ii) in Definition \ref{defn_5_4} 
so that we can use a similar preimage argument in the final portion of the argument. Thus we have the following 
theorem.

\begin{thm}
Any square-free morphism with unstackable image words is square-free.
\label{thm_5_5}
\end{thm}

\begin{proof}
Assume that $h$ is a square-free morphism with unstackable image words such that $|h(a)| = n$ for all $a\in\Sigma$. We must 
show that for some $W\in\Sigma^*$, $W$ is square-free if and only if $h(W)$ is square-free. We will begin 
with the easy direction.

\subsection{The $\Leftarrow$ direction} We will proceed by contrapositive. So assume that $W = AXXB$, 
where $X\in\Sigma^+$ and $A,B\in\Sigma^*$. Write
\[h(W) = h(AXXB) = h(A)h(X)h(X)h(B),\]
which contains the square $h(X)h(X)$. So we are done with this direction.

\subsection{The $\Rightarrow$ direction} Again we proceed by arguing the contrapositive. So we assume 
that
\begin{equation}
h(W) = Ac_{j_0}Xd_{i_0}c_{j_1}Xd_{i_1}B,
\label{eq_5_3}
\end{equation}
where $c=c_{j_0}=c_{j_1}\in\Sigma$, $d = d_{i_0} = d_{i_1}\in\Sigma$ and $A,X,B\in\Sigma^*$. Note that 
we are using $c_{j_0}$ and $c_{j_1}$ so that we can mark the beginning of the square, and similarly for 
the $d$'s and the end of the square.

There are two cases to consider here $|cXd|\not\equiv 0 \pmod{n}$ and $|cXd|\equiv 0 \pmod{n}$. We 
show 
that it is impossible for $|cXd|\not\equiv 0 \pmod{n}$ in an analogous manner as in Theorem 
\ref{thm_5_3}, 
as seen in section 5.1.2.1.

So assume that $|cXd|\equiv 0 \pmod{n}$. From the definition of the pooh square-free morphism we know 
that each image word for $h$ must begin and end with distinct letters. Further we know the suffix of the 
tile containing $c_{j_0}$ must be identical to the suffix of the tile containing $c_{j_1}$. Thus they are the 
same image word, call it $h(z)$ for some $z\in\Sigma$. So because
\[h(W) = AcXcXB,\]
we can find subwords $C,D,Y$ of $W$ such that
\[h(W) = h(CzYzYD) = AcXcXB.\]
Thus we have that $W = AzXzXB$ which contains a square, and we are done.
\end{proof}

\section{Acknowlegments}
I especially would like to thank Dr. George F. McNulty, my thesis advisor, for all
of his insights on the ideas presented here.

\bibliographystyle{plain}
\bibliography{RobsBibTexDatabase}

\end{document}